\documentclass{article}
\usepackage{enumitem}
\usepackage{ijcai23}

\usepackage{times}
\usepackage{soul}
\usepackage{url}
\usepackage[hidelinks]{hyperref}
\usepackage[utf8]{inputenc}
\usepackage[small]{caption}
\usepackage{graphicx}
\usepackage{amsmath}
\usepackage{amsthm}
\usepackage{booktabs}
\usepackage{algorithm}
\usepackage{algorithmic}
\usepackage[switch]{lineno}
\usepackage{multirow}
\usepackage{amssymb}

\usepackage{xcolor}
\usepackage{balance} \usepackage{xspace}

\definecolor{tucgreen}{RGB}{0,140,79}

\newcommand{\new}[1]{{\color{black}{#1}}}

\newcommand{\nmnew}[1]{{\color{black}{#1}}}
\newcommand{\WJnew}[1]{{\color{black}#1}}

\newcommand{\may}{\ensuremath{\mathit{may}}\xspace}
\newcommand{\must}{\ensuremath{\mathit{must}}\xspace}
\newcommand{\uval}{\mathbf{u}}   

\newtheorem{theorem}{Theorem}[section]
\newtheorem{corollary}[theorem]{Corollary}

\newtheorem{definition}[theorem]{Definition}
\newtheorem{remark}[theorem]{Remark}

\newcommand{\stratDiamond}[1]{\ensuremath{\langle\!\langle#1\rangle\!\rangle}}
\newcommand{\stratBox}[1]{\ensuremath{[\![#1]\!]}}

\newcommand{\atls}{ATL$^*$}
\newcommand{\ctls}{CTL$^*$}
\newcommand{\SL}{SL\xspace}
\newcommand{\APSet}{\textit{AP}}
\newcommand{\AgSet}{\textit{Ag}}
\newcommand{\AcSet}{\Act}
\newcommand{\StSet}{\States}
\newcommand{\trkElm}{\rho}
\newcommand{\SetN}{N}
\newcommand{\decFun}{\vec{\alpha}}
\newcommand{\DecSet}{\ActProfile}
\newcommand{\pthElm}{\pi}
\newcommand{\TrkSet}{\mathit{Trk}}
\newcommand{\PthSet}{\mathit{Pth}}
\newcommand{\sElm}{s}
\newcommand{\IS}{model}
\newcommand{\GName}{{G}}
\newcommand{\trnFun}{\trans}
\newcommand{\XSet}{X}

\newcommand{\strFun}{f}
\newcommand{\StrSet}{Str}
\newcommand{\defeq}{=}
\newcommand{\VarSet}{\textit{Var}}
\newcommand{\pto}{\rightarrow}
\newcommand{\asgFun}{\chi}
\newcommand{\X}{X}
\newcommand{\Until}{U}
\newcommand{\U}{\Until}
\newcommand{\F}{F}
\newcommand{\R}{R}
\newcommand{\xElm}{x}

\newcommand{\lElm}{l}
\newcommand{\EExs}[1]{\exists\:\! {#1}}
\newcommand{\AAll}[1]{\forall\:\!\!{#1}}
\newcommand{\aElm}{a}
\newcommand{\playElm}{\pi}
\newcommand{\AsgSet}{Asg}
\newcommand{\playFun}{\mathit{play}}
\newcommand{\plays}{\mathit{plays}}

\newcommand{\free}[1]{\textit{free}{(#1)}}

\newcommand{\LTL}{LTL}

\newcommand{\pElm}{p}

\newcommand{\G}{G}
\newcommand{\States}{\mathit{St}}
\newcommand{\state}{s}
\newcommand{\Act}{Act}
\newcommand{\Ag}{Ag}
\newcommand{\ActProfile}{ACT}

\newcommand{\trans}{\tau}

\newcommand{\set}[1]{\{#1\}}
\newcommand{\tuple}[1]{\langle#1\rangle}
\newcommand{\powerset}[1]{2^{#1}}
\newcommand{\abstr}{\mathcal{A}}

\title{Scalable Verification of Strategy Logic through Three-valued Abstraction}

\author{
Francesco Belardinelli$^1$
\and
Angelo Ferrando$^2$\and
Wojciech Jamroga$^{3}$\and
Vadim Malvone$^4$ \and
Aniello Murano$^5$
\affiliations
$^1$Imperial College London, United Kingdom \\
$^2$University of Genoa, Italy\\
$^3$SnT, University of Luxembourg \& Institute of Computer Science, Polish Academy of Sciences\\
$^4$Telecom Paris, France \\
$^5$University of Naples Federico II, Italy
\emails
francesco.belardinelli@imperial.ac.uk,
angelo.ferrando@unige.it,
w.jamroga@ipipan.waw.pl,
vadim.malvone@telecom-paris.fr,
aniello.murano@unina.it
}

\begin{document}

\maketitle

\begin{abstract}
    The model checking problem for multi-agent systems against Strategy Logic specifications is known to be non-elementary. On this logic several fragments have been defined to tackle this issue but at the expense of expressiveness. In this paper, we propose a three-valued semantics for Strategy Logic upon which we define an abstraction method. We show that the latter semantics is an approximation of the classic two-valued one for Strategy Logic. Furthermore, we extend MCMAS, an open-source model checker for multi-agent specifications, to incorporate our abstraction method and present some promising experimental results.	
\end{abstract}

\section{Introduction}
\label{sec:Introduction}

In multi-agent systems, logics for strategic reasoning play a key
role. In this domain, one of the success stories is Alternating-time Temporal Logic
(\atls)\cite{AHK02}, which can express cooperation and competition
among teams of agents in order to achieve temporal goals, such
as fairness, liveness, safety requirements.  In fact, \atls\ extends the well
known branching-time temporal logic \ctls\ \cite{Halpern+86a} by
generalizing the existential $E$ and universal $A$ path
quantifiers of \ctls\ with the strategic modalities $\stratDiamond{C}$
and $\stratBox{C}$, where $C$ is a coalition of agents. However, it has been
observed that \atls\ suffers from a number of limitations that, on the
one hand, make the model-checking and satisfiability problems
decidable (both are $2$ExpTime-complete); but, on the
other hand, make the logic too weak to express key game-theoretic
concepts, such as Nash equilibria
\cite{Mogavero12stratLogic-decidable}. To overcome these limitations,
\emph{Strategy Logic} (\SL) \cite{MogaveroMPV14,CHP10} has been put forward.
A key aspect of \SL\ is to consider strategies as first-order objects that can be
existentially or universally
quantified over by means of the
strategy quantifiers $\exists x$ and $\forall x$,
respectively. Then, by means of a binding operator $(a,x)$, a strategy
$x$ can be associated to a specific agent $a$. This allows to reuse
strategies as well as to share them among different agents.
Since its introduction, \SL\ has proved to be a powerful formalism: it
can express complex solution concepts, including Nash equilibria, and
subsumes all previously introduced logics for strategic reasoning,
including \atls. The high expressivity of \SL\ has spurred its
analysis in a number of directions and extensions, such as
prompt \cite{Aminof16promptATL-KR}, graded \cite{AMMR18}, fuzzy
\cite{BKMMMP19}, probabilistic \cite{AKMMR19}, and imperfect
\cite{BMMRV21,BLMR20} strategic reasoning.

As one may expect, the high expressivity of \SL\ comes at a
price. Indeed, its model-checking problem turns out to be
non-elementary \cite{MogaveroMPV14}.
Moreover, the model checking procedure is not immune to
the well-known state-space explosion, as faithful models of real-world
systems are intrinsically complex and often infeasible even to generate, let
alone verify.  These issues call for techniques to make model
checking \SL\ amenable at least in practice. A technique that has been
increasingly used in industrial settings to verify hardware and
software systems is state abstraction, which allows to reduce the
state space to manageable size by clustering ``similar'' concrete states
into abstract states. Abstraction has been first introduced for
stand-alone systems \cite{Clarke94abstraction}, then extended to
two-agent system verification
\cite{GrumbergLLS07,ShohamGrumberg04,Bruns+00a,AminofKM12}.
Recently, abstraction approaches have been investigated for
multi-agent systems w.r.t.~\atls\ specifications
\cite{KouvarosL17,BelardinelliLM19,BelardinelliL17,JamrogaKP16,JamrogaKKP20,DBLP:journals/ai/BelardinelliFM23,DBLP:conf/atal/Ferrando23,DBLP:journals/jair/BelardinelliLMY22,DBLP:conf/kr/BelardinelliLM18,DBLP:conf/kr/BelardinelliM20,DBLP:conf/atal/FerrandoM21,DBLP:conf/paams/FerrandoM22}.  A natural
direction is then to investigate a form of abstraction suitable
for Strategy Logic as well.

\paragraph{Our Contribution.} In this paper we introduce the first notion of three-valued abstraction for \SL.
The contribution of this paper is threefold. First, in Sec.~\ref{sec:3valued} we
define a three-valued semantics for \SL\ where, besides the standard
truth values true $\top$ and false $\bot$, we have a third value
undefined $\uval$ that models situations where the verification
procedure is not able to return a conclusive answer. Second, in Sec.~\ref{sec:abstraction} we
introduce an abstraction procedure for \SL, which can reduce
significantly the size of the state space of \SL\ models, although at
the cost of making some formulas undefined.  The main theoretical
result is the Preservation Theorem \ref{theor_pres}, which allows us to model
check \SL\ formulas in the three-valued abstraction and then lift any
defined answer to the original two-valued model.
Third, in Sec.~\ref{sec:experiments} we evaluate
empirically the trade-off between state-space reduction and
definiteness, by applying our abstraction procedure to a scheduling
scenario. What we observe empirically is a significant reduction of
the model size, which allows us to verify instances that are not amenable
to current model checking tools.

\paragraph{Related Work.}
The present contribution is inspired by a long tradition of works on
the abstraction of MAS models, including through three-valued
semantics.  An abstraction-refinement framework for the temporal logic
CTL over a three-valued semantics was first studied
in~\cite{ShohamGrumberg04,ShohamG07}, and then extended to the
full $\mu$-calculus \cite{GrumbergLLS07}
and 
hierarchical systems \cite{AminofKM12}.
Three-valued abstractions for the verification of Alternating-time Temporal Logic
have been put forward
in~\cite{BallKupferman06,LomuscioMichaliszyn14b,LomuscioMichaliszyn15,LomuscioMichaliszyn16a}.
In~\cite{BallKupferman06,ShohamGrumberg04} ATL$^*$ is interpreted under perfect information; while \cite{LomuscioMichaliszyn14b,LomuscioMichaliszyn15,LomuscioMichaliszyn16a} consider {\em non-uniform} strategies~\cite{RaimondiLomuscio05d}.
Finally, \cite{JamrogaKP16,JamrogaKKP20} introduce a multi-valued
semantics for ATL$^*$ that is a conservative extension of the
classical two-valued variant.
Related to this line, three-valued logics have been extensively applied to system verification, including \cite{Bruns99threevalued,HuthJagadeesanSchmidt01,Godefroid+03a}

Clearly, we build in this long line of works, but the expressiveness of
\SL\ raises specific challenges that the authors of the contributions
above need not to tackle. We briefly mention them here and refer to specific sections for further details. First, we have to introduce individual \must
and \may actions and strategies as under- and over- approximations of
the behaviours of our agents.  Second, the loosely-coupled nature of
agents requires to consider non-deterministic transitions in the
abstraction (Sec.~\ref{sec:abstraction}).  Third, the arbitrary alternation of existential and
universal strategy quantifiers makes proving the Preservation
Theorem~\ref{theor_pres} significantly more challenging, and complicates our
experiments in verifying three-valued \SL\ in the two-valued model-checking
tool MCMAS (Sec.~\ref{sec:experiments}).

\section{Reasoning about Strategies}\label{sec:SL}

In this section we recall the definitions of basic notions for Strategy Logic~\cite{MogaveroMPV14}.

\subsection{Syntax}

		\emph{Strategy Logic} (\emph{\SL})
syntactically extends \LTL\ with two \emph{strategy
                  quantifiers}, the existential $\EExs{\xElm}$ and
                universal $\AAll{\xElm}$, and an \emph{agent binding}
                $(\aElm, \xElm)$, where $\aElm$ is an agent and
                $\xElm$ a variable.  Intuitively, these additional elements
                can be respectively read as \emph{``there exists a
                  strategy $\xElm$''}, \emph{``for all strategies
                  $\xElm$''}, and \emph{``bind agent $\aElm$ to the
                  strategy associated with the variable $\xElm$''}.
Since negated quantifiers often prove problematic in many valued settings,
we restrict the syntax of \SL\ to formulas in Negation Normal Form (NNF), without loss of expressiveness. In that case, the universal strategic quantifier $\AAll{\xElm}$ and the temporal operator ``Release'' $R$ are added as primitives, and negation is allowed only at the level of literals. Note that every formula of \SL\ can be equivalently transformed to one in NNF, with at most a linear blowup \cite{MogaveroMPV14}.
\begin{definition}[\SL\ Syntax]
			\label{def:sl(syntax)}
			Given the set  $\APSet$ of atoms,
                        variables $\VarSet$, and agents $\AgSet$, the
                        formal syntax of \SL\ is defined as follows, 			where $\pElm \in \APSet$, $\xElm \in \VarSet$,
                        and $\aElm \in \AgSet$:
\begin{eqnarray*}
			\varphi &::=& \pElm \mid \neg \pElm 
                        \mid \varphi \wedge \varphi
                        \mid \varphi \vee \varphi
				        \mid \allowbreak \EExs{\xElm} \varphi 
				        \mid \AAll{\xElm} \varphi 
				        \mid (\aElm, \xElm) \varphi \mid \\
                    &&    \X \varphi 
                        \mid \varphi \:\Until \varphi
                        \mid \varphi \,\R\, \varphi
			\end{eqnarray*}
\end{definition}

We introduce the derived temporal operators as usual: $\F\varphi =  \top \U \varphi$ (``eventually'') and $\G\varphi =  \bot \R \varphi$ (``always'').

		Usually, predicate logics need the concepts of free and bound
		\emph{placeholders} in order to formally define their semantics.
In \SL, since strategies can be associated to both
                agents and variables, we introduce the set of \emph{free
                  agents/variables} $\free{\varphi}$ as the subset of
                $\AgSet \cup \VarSet$ containing \emph{(i)} all agents
                $\aElm$ for which there is no binding $(\aElm, \xElm)$
                before the occurrence of a temporal operator, and
                \emph{(ii)} all variables $\xElm$ for which there is a
                binding $(\aElm, \xElm)$ but no quantification
                $\EExs{\xElm}$ or $\AAll{\xElm}$.  A formula $\varphi$
                without free agents (resp., variables), i.e., with
                $\free{\varphi} \cap \AgSet = \emptyset$ (resp.,
                $\free{\varphi} \cap \VarSet = \emptyset)$, is called
                \emph{agent-closed} (resp., \emph{variable-closed}).
                If $\varphi$ is both agent- and variable-closed, it is a \emph{sentence}.

\subsection{Two-valued Semantics}

We now provide a formal semantics to Strategy Logic.

\paragraph{Models.}
\!To model the behaviour of multi-agent systems, we use a variant of concurrent game structures~\cite{AHK02}.
\begin{definition}[CGS] \label{def:cgs}
	A \emph{concurrent game structure
(CGS)} is a tuple $G = \langle \Ag, \States, \state_0, \Act,
\trans, AP, V \rangle$ such that (i)
$\Ag$ is a finite, non-empty set of {\em agents}.
(ii) $\States$ is a finite, non-empty set of {\em states}, with {\em initial state}  $\state_0 \in \States$.	
(iii)  $\Act$ is a finite, nonempty  set of {\em actions}.
		We use $\ActProfile = \Act^{|\Ag|}$ for the set of all \emph{joint actions} (a.k.a.~\emph{action profiles}), i.e., tuples of individual actions, played synchronously by all agents.	
(iv) $\trans : \States \times \ActProfile \to \powerset{\States}$ is the \emph{transition function} assigning
		successor states $\{\state',\state'',\dots\} = \trans(\state, \vec{\alpha})$ to each
		state $\state\in \States$ and joint action $\vec{\alpha} \in \ActProfile$.
        We assume that the transitions in a CGS are deterministic, i.e., $\trans(\state, \vec{\alpha})$ is always a singleton.\footnote{
          The deterministic transitions in a CGS are usually defined by a function of type $\trans : \States \times \ActProfile \to \States$.
          We use a slightly different (but equivalent) formulation. This will make it easier for us to extend it to nondeterministic transitions in three-valued models (see Def.~\ref{def:3CGS}).}
(v) $AP$ is a set of \emph{atomic propositions}, and (vi) $V: \States \times AP \rightarrow \{\top,\bot\}$
		is a {\em two-valued labelling function}.
\end{definition}

By Def.~\ref{def:cgs} a CGS describes the interactions of a group $\Ag$ of
agents, starting from the initial state $\state_0 \in \States$, according to the
transition function $\trans$.
We use $\GName$ as a subscript for $\AgSet_\GName$, $\States_\GName$, etc., whenever the model is not clear from the context.

Note that the CGSs used in the semantics of Strategy Logic assume that all the actions are available to every agent at every state~\cite{MogaveroMPV14}. This is because a strategy assigned to variable $\xElm$ can be later associated with any agent $\aElm$ by means of the binding operator $(\aElm, \xElm)$. 
As a consequence, the available strategies (and hence also the available actions) 
are the same for every agent.

\paragraph{Tracks, Paths, Strategies.}\label{sec:paths}
\!\!\new{We denote the $i$-th element of a tuple $v$ as $v_i$, the prefix of $v$ of lenght $i$ as $v_{\leq i}$, and with $last(v)$ as the last element of $v$.}
A \emph{track} is a {\em finite nonempty} sequence of states $\trkElm \in \StSet^{+}$
such that, for all $0 \le i \le |\trkElm| - 1$,
there is an action profile $\decFun \in \DecSet$ with $(\trkElm)_{i + 1} \in \trnFun((\trkElm)_{i}, \decFun)$.
Similarly, a \emph{path} is an {\em infinite} sequence of states $\pthElm \in \StSet^{\omega}$
such that, for all $i \in \SetN$,
there is $\decFun \in \DecSet$ with $(\pthElm)_{i + 1} \in \trnFun((\pthElm)_{i}, \decFun)$.
The set $\TrkSet \subseteq \StSet^{+}$ contains all the tracks in the model, and $\TrkSet(\sElm)$ the tracks starting at state $\sElm \in \StSet$.
The sets $\PthSet$ and $\PthSet(\sElm)$ are defined analogously.
We denote the prefix of a path $\pi$ up to position $i \in \mathbb{N}$ as $\pi_{\leq i}$.

A \emph{strategy} is a partial function $\strFun :
              \TrkSet \pto \AcSet$ that maps each 
track in its domain to an action.
Intuitively, a strategy is a conditional plan that,
              for some tracks of $\GName$, prescribes an
              action to be executed.
A strategy is {\em memoryless} (or positional), if $last(\rho) = last(\rho')$ implies $\strFun(\rho) = \strFun(\rho')$, that is, the strategy only depends on the last state.
The set $\StrSet \defeq \TrkSet \pto \AcSet$ (resp.,
              $\StrSet(\sElm) \defeq \TrkSet(\sElm) \to \AcSet$)
              contains all 
strategies \new{(resp., strategies starting from $s$)}.

\paragraph{Assignments.}
\!\!Let $\VarSet$ be the set of variables.  An
                \emph{assignment}
is a partial
                function $\asgFun : \VarSet \cup \AgSet \pto \StrSet$
                mapping variables and agents in its domain to
                strategies.
An assignment $\asgFun$ is \emph{complete} if it is
                defined on all agents, i.e., $\AgSet \subseteq
                dom(\asgFun)$.  
The set $\AsgSet \defeq \VarSet \cup \AgSet \pto
                \StrSet$ 
contains all 
assignments.  More\-over, $\AsgSet(\XSet)
                \!\defeq \!\XSet \!\to\! \StrSet$ 
indicates the
                subset of \emph{$\XSet$-defined} 
assignments, i.e., 
assignments defined on $\XSet \subseteq
                \VarSet \cup \AgSet$.

		As in first-order logic, in order to quantify over
                strategies or bind a strategy to an agent, we update
                an assignment $\asgFun$ by associating an agent or a
                variable $\lElm$ with a new strategy $\strFun$.  Let
                $\asgFun \in \AsgSet$ be an assignment, $\strFun \in
                \StrSet$ a strategy and $\lElm \in \VarSet \cup
                \AgSet$ either an agent or a variable.  Then,
                $\asgFun[\lElm \mapsto \strFun] \in \AsgSet$ denotes
                the new assignment 
that returns $\strFun$ on $\lElm$ and the
                same value that $\asgFun$ would return on the rest of
                its domain.

\paragraph{Outcome Plays of a Strategy.}
\!\!A \emph{play} is the unique outcome of the game
                settled by all agent strategies engaged in
                it. Formally, given a state $\sElm \!\in\! \StSet$ and a
                               complete 
assignment $\asgFun \in
                \AsgSet(\sElm)$, the function $\playFun(\asgFun,
                \sElm)$ returns the path $\playElm \in \PthSet(\sElm)$
                such that, for all $i \in \SetN$, it holds that
                $\{\playElm_{i + 1}\} = \trnFun(\playElm_{i},
                \decFun)$, where $\decFun(\aElm) \defeq
                \asgFun(\aElm)(\playElm_{\leq i})$ for each $\aElm
                \in \AgSet$.

We now define the translation of an
assignment together with a related path (resp.~state).
It is used to keep track, at a certain stage of the play, of the current state and its updated assignment. 
For a path $\playElm$ and an 
assignment $\asgFun \in \AsgSet$, the \emph{$i$-th global translation} of $(\asgFun,
\playElm)$ with $i \in \SetN$ is the pair $(\asgFun, \playElm)^{i} \defeq
 (\asgFun_{\playElm_{\leq i}}, \playElm_{i})$ of an assignment and a state.
Moreover, for a state $\sElm \in \StSet$, we define
$(\asgFun, \sElm)^{i} \defeq(\asgFun, \playFun(\asgFun, \sElm))^{i}$.

		As in the case of components of a \IS, in order to
                avoid any ambiguity, we sometimes use the name of the
                \IS\ as a subscript of the sets and functions
                introduced above.

\paragraph{Satisfaction.}
The (two-valued) satisfaction relation for \SL\ is defined as follows.
		\begin{definition}[Two-valued Satisfaction]
\label{def:sl(semantics)}
			Given a \IS\ $\GName$, for all \SL\ formulas
                        $\varphi$, states $\sElm \in \StSet$, and
assignments $\asgFun \in
                        \AsgSet$ with $\free{\varphi} \subseteq
                        dom(\asgFun)$, the satisfaction relation $(\GName,
                        \asgFun, \sElm) \models^2 \varphi$ is
                        inductively defined as follows:
\begin{itemize}[leftmargin=15pt]
				\item\label{def:sl(semantics:ap)}
                                  $(\GName, \asgFun, \sElm) \models^2
                                  \pElm$ iff $V(\sElm, \pElm ) =
                                  \top$, for $\pElm \in \APSet$.
                                  
				\item\label{def:sl(semantics:bool)}
                                  Boolean operators are interpreted as
                                  usual.
\item\label{def:sl(semantics:eqnt)}
                                                  $(\GName, \asgFun,
                                                  \sElm) \models^2
                                                  \EExs{\xElm}
                                                  \varphi$ iff for some
strategy \new{$\strFun
                                                  \in \StrSet(s)$},
                                                  $(\GName,
                                                  \asgFun[\xElm
                                                    \mapsto \strFun],
                                                  \sElm) \models^2
                                                  \varphi$.
						\item\label{def:sl(semantics:aqnt)}
							$(\GName,
                                  \asgFun, \sElm) \models^2
                                  \AAll{\xElm} \varphi$ iff for all
strategies \new{$\strFun
                                  \in \StrSet(s)$},
$(\GName,{ \asgFun[\xElm \mapsto
                                    \strFun]}, \sElm) \models^2
                                  \varphi$.
\item\label{def:sl(semantics:bnd)} 
$(\GName, \asgFun, \sElm)
                                  \models^2 (\aElm, \xElm) \varphi$ iff
                                  $(\GName, \asgFun[\aElm \mapsto
                                    \asgFun(\xElm)], \sElm) \models^2
                                  \varphi$.
				\item\label{def:sl(semantics:path)}
                                  Finally, if the assignment $\asgFun$
                                  is also complete, 
it holds that:
					\begin{itemize}[leftmargin=10pt]
						\item\label{def:sl(semantics:next)}
                                                  $(\GName, \asgFun,
                                                  \sElm) \models^2 \X
                                                  \varphi$ iff $(\GName,
                                                  (\asgFun, \sElm)^{1})
                                                  \models^2 \varphi$;
						\item\label{def:sl(semantics:until)}
                                                  $(\GName, \asgFun,
                                                  \sElm) \models^2
                                                  \varphi_{1} \Until
                                                  \varphi_{2}$ iff
                                                  for some index $i
                                                  \in \SetN$, 
                                                  $(\GName, (\asgFun,
                                                  \sElm)^{i}) \models^2
                                                  \varphi_{2}$ and,
                                                  for all $j < i$, it
                                                  holds that $(\GName,
                                                  (\asgFun, \sElm)^{j})
                                                  \models^2
                                                  \varphi_{1}$;
						\item\label{def:sl(semantics:release)}
                                $(\GName, \asgFun, \sElm) \models^2 \varphi_{1} \R \varphi_{2}$ iff, 
for all $i \in \SetN$, $(\GName, (\asgFun,\sElm)^{i}) \models^2 \varphi_{2}$ or
                                there is $j\le i$ such that $(\GName, (\asgFun, \sElm)^{i}) \models^2 \varphi_{1}$.
\end{itemize}
			\end{itemize}
		\end{definition}

Due to the semantics of the Next $X$,
                Until $U$, and Release $R$ operators,
\LTL\ semantics is clearly embedded into the \SL\ one.
                Furthermore, since the satisfaction of a sentence
                $\varphi$ does not depend on assignments, we omit them
                and write $(\GName, \sElm) \models \varphi$, when
                $\sElm$ is a generic state in $\StSet$, and $\GName
                \models \varphi$ when $\sElm = \sElm_0$.

Note that we can easily define the memoryless variant of strategy logic by restricting the clauses for operators $\EExs{\xElm}$ and $(\aElm, \xElm)$ to memoryless strategies.

Finally, we define the (two-valued) model checking problem for \SL\ as determining whether an \SL\ formula $\phi$ holds in a CGS $G$, that is, whether $G \models^2 \phi$. We conclude this section by stating the related complexity result.
\begin{theorem}[\cite{MogaveroMPV14}]
    The model checking problem for Strategy Logic is non-elementary.           
\end{theorem}

\section{Three-Valued Strategy Logic}\label{sec:3valued}

In this section we introduce a novel three-valued semantics for Strategy Logic starting by extending CGSs.

\subsection{Three-Valued CGSs}

We extend (two-valued) CGSs with \must and \may transitions as under- and over-approximations of the strategic abilities of agents.
\begin{definition}[Three-valued CGS]\label{def:3CGS}
\mbox{\!A \emph{three-valued CGS} 
is a} tuple $\GName \!=\!\tuple{\Ag, \States, \state_0, \Act^{\may}\!, \Act^{\must}\!,
\trans^{\may}\!, \trans^{\must}\!, AP, V}$, where:
\begin{itemize}[leftmargin=15pt]
\item $\Ag, \States, \state_0, AP$ are defined as 
in Def.~\ref{def:cgs}.

\item $\Act^{\may}$ and $\Act^{\must}$ provide respectively the upper and lower approximation of the available actions. We assume that $\Act^{\must} \subseteq \Act^{\may}$. The sets of \may and \must action profiles are given by
    $\ActProfile^{\may} = (\Act^{\may})^{|\Ag|}$ and $\ActProfile^{\must} = (\Act^{\must})^{|\Ag|}$, respectively.

\item $\trans^{\may}: \States \times \ActProfile^{\may} \rightarrow \powerset{\States}$ is the \may transition function,
and $\trans^{\must}: \States \times \ActProfile^{\may} \rightarrow \powerset{\States}$ the \must transition function.
Note that both functions are possibly nondeterministic and are \emph{defined} on all the {potential} action profiles in the system, i.e., $\ActProfile^{\may}$.
However, we only require that they \emph{return nonempty successor sets} on their respective action profiles.
That is, $\trans^{\may}(\state,\vec{\alpha}) \neq \emptyset$ for every state $\state\in\States$ and action profile $\vec{\alpha}\in\ActProfile^{\may}$, and
$\trans^{\must}(\state,\vec{\alpha}) \neq \emptyset$ for every state $\state\in\States$ and action profile $\vec{\alpha}\in\ActProfile^{must}$.\footnote{Note that the function $\trans^{\must}$ is total because we assume the empty set as an element of the co-domain.} Moreover, it is required that $\trans^{\must}(\state,\vec{\alpha}) \subseteq \trans^{\may}(\state,\vec{\alpha})$ for every $\state\in\States$ and $\vec{\alpha}\in\ActProfile^{\may}$. In other words, every \must transition is
also a \may transition, but not necessarily viceversa.

\item The labelling function $V : \StSet \times \APSet \to \set{\bot, \top, \uval}$ maps now each pair of a state and an atom to a truth value of ``true,'' ``false,'' or ``undefined.''
\end{itemize}
\end{definition}

The notions of tracks, paths, and the definitions of sets $\TrkSet,\TrkSet(\sElm),\PthSet,\PthSet(\sElm)$ carry over from Section~\ref{sec:paths}.

\paragraph{May/Must Strategies and their Outcomes.}

A \emph{$\may$-strategy} (resp.~\emph{$\must$-strategy}) is a function $\strFun : \TrkSet \pto \AcSet^{\may}$ (resp.~$\AcSet^{\must}$)
that maps each 
track 
to a $\may$ (resp.~$\must$) action. Note that each $\must$-strategy is a $\may$-strategy, but not necessarily the other way around. Moreover, we can define memoryless \may- and \must-strategies in the standard way.
The sets $\StrSet^{\may}$ and 
$\StrSet^{\must}$
are defined analogously to Section~\ref{sec:paths}.

Given a state $\state\in\States$ and a profile of (\may and/or \must) strategies, represented by a complete 
assignment $\asgFun \in \AsgSet$, we define two kinds of outcome sets, $\plays^{\may}(\asgFun,\state)$ and $\plays^{\must}(\asgFun,\state)$. The former over-approximates the set of paths that can really occur when executing $\asgFun$ from $\state$, while the latter under-approximates it.
Typically, we will use $\plays^{\may}$ to establish that the value of a temporal formula $\varphi$ is $\top$ (if $\varphi$ holds in all such paths), and $\plays^{\must}$ 
for $\bot$ (if $\varphi$ is false in at least one path).
Formally, the function $\plays^{\may}(\asgFun,\sElm)$ returns the paths $\playElm \in \PthSet(\sElm)$
such that, for all $i \in \SetN$, it holds that $\playElm_{i + 1} \in \trnFun^{\may}(\playElm_{i},\decFun)$,
where $\decFun(\aElm) \defeq$ $ \asgFun(\aElm)(\playElm_{\leq i})$ for each $\aElm \in \AgSet$.
The definition of $\plays^{\must}(\asgFun,\sElm)$ is analogous, only with $\trnFun^{\must}$ being used instead of $\trnFun^{\may}$.

\subsection{Three-valued Semantics}

We now define the Three-valued satisfaction relation for Strategy Logic.
\begin{definition}[Three-valued Satisfaction]
			\label{def:sl3(semantics)}
			Given a 3-valued \IS\ $\GName$, for all
                        \SL\ formulas $\varphi$, states $\sElm \in
                        \StSet$, and 
assignments
                        $\asgFun \!\in \! \AsgSet(\sElm)$ with
                        $\free{\varphi}\! \subseteq \! dom(\asgFun)$, the
                        satisfaction relation $(\GName, \asgFun, \sElm
                        \models^3 \varphi) \! = \! tv$ is inductively
                        defined as follows.
\begin{itemize}[leftmargin=9pt]
\item\label{def:sl3(semantics:ap)}
                              $(\GName, \asgFun, \sElm \models^3
                              \pElm)\ =\ V(\sElm,
                              \pElm)$, for $\pElm \in
                              \APSet$.
\item\label{def:sl3(semantics:bool)}
                              Boolean operators are interpreted as
                              in {\L}ukasiewicz's three valued logic~\cite{Lukasiewicz20logic}. 
\item\label{def:sl3(semantics:qnt)} 
For $\phi = \exists x \varphi$,
\begin{itemize}[leftmargin=8pt]
  \item\label{def:sl3(semantics:eqnt)}
    $(\GName, \asgFun, \sElm \models^3 \phi) = \top$ iff
    $(\GName, \asgFun[\xElm \mapsto \strFun], \sElm \models^3 \varphi) =\top$
    for some 
$\must$-strategy \new{$\strFun \in \StrSet^{\must}(s)$};
  \item $(\GName, \asgFun, \sElm \models^3 \phi) = \bot$ iff
    $(\GName, \asgFun[\xElm \mapsto \strFun], \sElm \models^3 \varphi) = \bot$
    for all 
$\may$-strategies $\strFun \in \StrSet^{\may}(\sElm)$;
  \item otherwise, $(\GName, \asgFun, \sElm \models^3 \phi) = \uval$.
\end{itemize}
\item For $\phi = \forall x \varphi$,
  \begin{itemize}[leftmargin=8pt]
  \item\label{def:sl3(semantics:aqnt)}
                                                  $(\GName, \asgFun,
                                                  \sElm \models^3
                                                  \phi) = \top$ iff
                                                  for all
$\may$-strategies
                                                  \new{$\strFun \in
                                                  \StrSet^{\may}(s)$},
                                                  $(\GName,
                                                  \asgFun[\xElm
                                                    \mapsto \strFun],
                                                  \sElm \models^3
                                                  \varphi) = \top$;
						\item $(\GName,
                                                  \asgFun, \sElm
                                                  \models^3
                                                  \phi) = \bot$ iff
                                                  for some
$\must$-strategy
                                                  \new{$\strFun \in
                                                  \StrSet^{\must}(s)$},
                                                  $(\GName,
                                                  \asgFun[\xElm
                                                    \mapsto \strFun],
                                                  \sElm \models^3
                                                  \varphi) = \bot$;
  \item otherwise, $(\GName, \asgFun, \sElm \models^3 \phi) = \uval$.
  \end{itemize}
\item\label{def:sl3(semantics:bnd)} 
$(\GName, \asgFun, \sElm \models^3 (\aElm, \xElm) \varphi)\ =\
   (\GName, \asgFun[\aElm \mapsto \asgFun(\xElm)], \sElm \models^3 \varphi)$.
\item\label{def:sl3(semantics:path)}
                              Finally, if the assignment $\asgFun$
                              is also complete, 
we define:
\begin{itemize}[leftmargin=8pt]
\item\label{def:sl3(semantics:next)}
  $(\GName, \asgFun, \sElm \models^3 \X \varphi) = \top$ iff
  for all $\playElm \in \plays^{\may}(\asgFun,\sElm)$, 
  we have $(\GName, (\asgFun, \playElm)^{1} \models^3 \varphi) = \top$;
\item
  $(\GName, \asgFun, \sElm \models^3 \X \varphi) = \bot$ iff 
  for some $\playElm \in \plays^{\must}(\asgFun,\sElm)$, 
  we have $(\GName, (\asgFun, \playElm)^{1} \models^3 \varphi) = \bot$;
\item
  otherwise, $(\GName, \asgFun, \sElm \models^3 \X \varphi) = \uval$.

\smallskip
\item\label{def:sl3(semantics:until)}
  $(\GName, \asgFun,\sElm \models^3 \varphi_{1} \Until \varphi_{2}) = \top$ iff
  for all $\playElm \in \plays^{\may}(\asgFun,\sElm)$, 
  there is $i \in \SetN$ such that 
  $(\GName, (\asgFun, \playElm)^{i} \models^3 \varphi_{2}) = \top$,
  and for all $j < i$ we have $(\GName, (\asgFun, \playElm)^{j} \models^3 \varphi_{1}) = \top$;
\item
  $(\GName, \asgFun, \sElm \models^3 \varphi_{1} \Until \varphi_{2}) = \bot$ iff
  for some $\playElm \in \plays^{\must}(\asgFun,\sElm)$ 
  and all $i \in \SetN$, 
  we have $(\GName, (\asgFun, \playElm)^{i} \models^3 \varphi_{2}) = \bot$
  or there exists $j < i$  such that $(\GName, (\asgFun, \playElm)^{j} \models^3 \varphi_{1}) = \bot$;
\item
  otherwise, $(\GName, \asgFun, \sElm \models^3 \varphi_{1} \Until \varphi_{2}) = \uval$.

\smallskip
\item\label{def:sl3(semantics:release)} 
  $(\GName, \asgFun, \sElm \models^3 \varphi_{1} \R \varphi_{2}) = \top$ if
  for all $\playElm \in \plays^{\may}(\asgFun,\sElm)$ 
  and $i \in \SetN$, 
  we have $(\GName, (\asgFun, \playElm)^{i} \models^3 \varphi_{2}) = \top$
  or there exists $j \le i$ such that $(\GName, (\asgFun, \playElm)^{j} \models^3 \varphi_{1}) = \top$;

\item
  $(\GName, \asgFun, \sElm \models^3 \varphi_{1} \R \varphi_{2}) = \bot$ iff
  for some $\playElm \in \plays^{\must}(\asgFun,\sElm)$ 
  and $i \in \SetN$, 
  we have $(\GName, (\asgFun, \playElm)^{i} \models^3 \varphi_{2}) = \bot$
  and for all $j \le i$, we have $(\GName, (\asgFun, \playElm)^{j} \models^3 \varphi_{1}) = \bot$;

\item
  otherwise, $(\GName, \asgFun, \sElm \models^3 \varphi_{1} \R \varphi_{2}) = \uval$.
\end{itemize}
\end{itemize}
\end{definition}

Again, we can define the memoryless, three-valued satisfaction relation for \SL
by restricting the clauses for operators $\EExs{\xElm}$, $\forall x$, and $(\aElm, \xElm)$ to memoryless strategies.
Similarly to Section~\ref{sec:SL}, if $\phi$ is a sentence, then $(\GName, \sElm \models^3 \varphi) = (\GName, \asgFun, \sElm \models^3 \varphi)$ for any assignment $\asgFun$, and $(\GName \models^3 \varphi) = (\GName, \sElm_0 \models^3 \varphi)$.

We now show that our three-valued semantics in Def.~\ref{def:sl(semantics)} is a conservative extension of the standard two-valued interpretation in Sec.~\ref{sec:SL}.
\begin{theorem}[Conservativeness] \label{lemma1}
	\label{c3v}
	Let $\G$ be a standard CGS, that is, 
$\AcSet^{\may} = \AcSet^{\must}$,
	$\trans^{\\may} = \trans^{\\must}$ are functions,
	and the truth value of every atom is defined (i.e., it is equal to either $\top$ or $\bot$).
	Then, for every formula $\phi$ in SL,
	\begin{eqnarray} (\G, \chi, \state \models^3 \phi ) = \top
		& \Leftrightarrow & (\G, \chi, \state) \models^2 \phi \label{eq1} \\
		(\G, \chi, \state \models^3 \phi ) = \bot & \Leftrightarrow & (\G, \chi,
		\state) \not \models^2 \phi \label{eq2}
	\end{eqnarray}
\end{theorem}

\begin{proof}
The result follows from the fact that in standard CGS the clauses for the three-valued satisfaction relation collapse to those for two-valued satisfaction, whenever 
$\AcSet^{\may} = \AcSet^{\must}$,
$\trans^{\\may} = \trans^{\\must}$ are functions,
	and the truth value of every atom is defined.
\end{proof}

\begin{remark}[Model checking]
For any syntactic fragment $\mathcal{L}$ of \SL, model checking of $\mathcal{L}$ with 3-valued semantics can be reduced to 2-valued model checking of $\mathcal{L}$ by a construction similar to~\cite[Theorem~4]{JamrogaKP16}. 
Note also that 2-valued model checking for $\mathcal{L}$ is a special case of its 3-valued counterpart, due to Theorem~\ref{lemma1}.
Thus, the decidability and complexity for 2-valued model checking in fragments of \SL carry over to 3-valued verification.
\end{remark}
 
﻿\section{Three-valued Abstraction for \SL} \label{sec:abstraction}

Here, we define the 3-valued state abstraction for CGS. The idea is to cluster the states of a CGS (called the \emph{concrete model}) according to a given equivalence relation $\approx$, e.g., one provided by a domain expert. 
\new{Typically, two states are deemed equivalent if they agree on the evaluation of atoms, possibly just the atoms appearing on a given formula $\phi$ to be checked. In some cases, such an equivalence relation might be too coarse and therefore more domain-dependent information could be taken into account.} 

Then, the sets of may (resp. must) actions and the may (resp. must) transitions are computed in such a way that they always overapproximate (resp. underapproximate) the actions and transitions in the concrete model. Formally, the abstraction is defined as follows.

\begin{definition}[Abstraction]\label{def:abstraction}
\!Let $\GName = \langle \Ag, \States, \state_0, \Act, \trans, $ $ AP, V \rangle$ be a CGS, and $\approx\subset\States\times\States$ an equivalence relation. We write $[\state]$ for the equivalence class of $\approx$ that contains $\state$.
The \emph{abstract model of $G$} w.r.t.~$\approx$ is defined as the 3-valued CGS
$\abstr(\GName) = \langle \abstr(\Ag), \abstr(\States), \abstr(\state_0), \abstr^{\may}(\Act), \abstr^{\must}(\Act), \abstr^{\may}(\trans),\linebreak \abstr^{\must}(\trans), \abstr(AP), \abstr(V) \rangle$, with:
\begin{itemize}[leftmargin=15pt]
\item
  $\abstr(\Ag) = \Ag$ and  $\abstr(AP) = AP$.

\item
  $\abstr(\States) = \{[\state] \mid \state\in\States \}$ with
$\abstr(\state_0) = [\state_0]$.
\item
  $\abstr^{\may}(\Act) = \Act$.
\item
  $\abstr^{\may}(\trans) = \trans^{\may} : \abstr(\States) \times (\abstr^{\may}(\Act))^{|\abstr(\Ag)|} \to \powerset{\abstr(\States)}$ such that\qquad
  $\trans^{\may}([\state],\vec{\alpha}) =$ \\
  \centerline{$\{[\state_{succ}] \mid \exists \state'\in[\state]\ \exists \state'_{succ}\in[\state_{succ}]\ .\  \state'_{succ}\in\trans(\state',\vec{\alpha}) \}$).}
\item
  $\abstr^{\must}(\trans) = \trans^{\must} : \abstr(\States) \times (\abstr^{\may}(\Act))^{|\abstr(\Ag)|} \to \powerset{\abstr(\States)}$ such that\qquad
  $\trans^{\must}([\state],\vec{\alpha}) =$ \\
  \centerline{$\{[\state_{succ}] \mid \forall \state'\in[\state]\ \exists \state'_{succ}\in[\state_{succ}]\ .\  \state'_{succ}\in\trans(\state',\vec{\alpha}) \}$).}
\item
  $\abstr^{\must}(\Act)$ is a maximal\footnote{with respect to set inclusion.} set $\Act^{\must} \subseteq \Act$ such that
  $\forall \state\in\States\ \forall\vec{\alpha} \in (\Act^{\must})^{|\abstr(\Ag)|}\ .\ \trans^{\must}([\state],\vec{\alpha}) \neq \emptyset$.
  Note that a unique maximal set does not always exist. In such cases, a natural heuristics would be to choose the maximal subset of actions with the largest cardinality, breaking ties lexicographically in case there are still multiple solutions.

\item
  $\abstr(V)([\state],p) = \left\{
                        \begin{array}{ll}
                        \top & \text{if } V(\state',p)=\top\text{ for all }s'\in[\state] \\
                        \bot & \text{if } V(\state',p)=\bot\text{ for all }s'\in[\state] \\
                        \uval & \text{otherwise.}
                        \end{array}
                        \right.$
\end{itemize}
\end{definition}

Note that $\abstr(\GName)$ can be computed in polynomial time w.r.t. the size of $\GName$, assuming the above heuristics for $\abstr^{\must}(\Act)$.

\WJnew{We now prove that the abstraction preserves classical truth values.
Given a strategy $f$ in $\GName$, we define the set of corresponding \may-strategies in $\abstr(\GName)$ by $abstr^\may(f) = \set{f^{\dagger} \mid f^{\dagger}([s_0],\ldots, [s_n]) = f(s_0',\ldots, s_n')\text{ for some }s'_0 \in [s_0], \dots, s'_n \in [s_n]}$.
Moreover, $abstr^\must(f) = abstr^\may(f) \cap \StrSet^{\must}$. 
Note that $abstr^\may(f)$ is always nonempty. Also, $abstr^\must(f)$ is either empty or a singleton.

Conversely, given a (\may or \must) strategy $f$ in $\abstr(\GName)$, we define the set of corresponding concrete strategies in $\GName$ by 
$concr(f) = \set{ f^* \mid f^*(s_0,\ldots, s_n) = f([s_0],\ldots, [s_n]) }$.
Notice that $concr(f)$ is always a singleton for \must strategies, and either empty or a singleton for \may strategies.
We lift $abstr^\may, abstr^\must, concr$ to sets of strategies in the standard way.
Clearly, $f \in concr(abstr^\may(f))$ for any concrete strategy, and $f \in abstr^\must(concr(f))$ for any \must-strategy.
We lift the notation to assignments analogously.
Observe that, in every $\asgFun^* \in concr(\asgFun[\xElm \mapsto \strFun])$, 
$\xElm$ is assigned with $\strFun^* \in concr(\strFun)$.
}
\begin{theorem}[Preservation] \label{theor_pres}
Let $\G$ be a CGS and $\abstr(\GName)$ its abstraction induced by equivalence relation $\approx$.
Then, for every formula $\phi$ in SL, \WJnew{every (\may or \must) assignment $\chi$} and state $s$ in $\abstr(\GName)$, \WJnew{every assignment $\chi^*\in concr(\chi)$}, and state $t \in s$ in $\G$, it holds that:
	\begin{eqnarray}
	((\abstr(\GName), \chi, \state) \models^3 \phi ) = \top
		& \Rightarrow & (\G,  \chi^*, t) \models^2 \phi \label{eq3} \\
	((\abstr(\GName),  \chi, \state) \models^3 \phi ) = \bot
		& \Rightarrow & (\G,  \chi^*, t) \not \models^2 \phi  \label{eq4}
	\end{eqnarray}
\end{theorem}

\begin{proof}
The proof is by induction on the structure of $\phi$.

\smallskip\noindent
\underline{Induction base ($\phi = p$):} 
$((\abstr(\GName),  \chi, \state) \models^3 \phi ) = \top$ iff $\abstr(V)(s, p) = \top$, iff for all $t \in s$, $V(t, p) = \top$, that is, $(\G,  \chi^*, t) \models^2 \phi$. The case for $\bot$ is proved similarly. \new{The case of \underline{$\phi=\lnot p$} is analogous.}

\smallskip\noindent
\underline{Case $\phi = \psi_1 \land \psi_2$:} 
$((\abstr(\GName),  \chi, \state) \models^3 \phi ) = \top$ iff
$((\abstr(\GName),  \chi, \state) \models^3 \psi_1 ) = \top$ and $((\abstr(\GName),  \chi, \state) \models^3 \psi_2 ) = \top$. 
By induction, for all \WJnew{$\chi^*\in concr(\chi)$ and} $t \in s$, $(\G,  \chi^*, t) \models^2 \psi_1$ and $(\G,  \chi^*, t) \models^2 \psi_2$. Thus, $(\G,  \chi^*, t) \models^2 \psi_1 \land \psi_2$. 

Further, $((\abstr(\GName),  \chi, \state) \models^3 \phi ) = \bot$ iff
$((\abstr(\GName),  \chi, \state) \models^3 \psi_1 ) = \bot$ or $((\abstr(\GName),  \chi, \state) \models^3 \psi_2 ) = \bot$.
By induction, for all \WJnew{$\chi^*\in concr(\chi)$ and} $t \in s$, $(\G,  \chi^*, t) \not \models^2 \psi_1$ or for all \WJnew{$\chi^*\in concr(\chi)$ and} $t \in s$, $(\G,  \chi^*, t) \not \models^2 \psi_2$. 
Thus, for all \WJnew{$\chi^*\in concr(\chi)$ and} $t \in s$, $(\G,  \chi^*, t) \not \models^2 \psi_1 \land \psi_2$. 
The case of \underline{$\phi = \psi_1 \lor \psi_2$} is analogous.

\smallskip\noindent
\underline{Case $\phi = \EExs{\xElm} \psi$:} 
$(\abstr(\GName), \asgFun, \sElm \models^3 \phi) = \top$ iff for some 
$\must$-strategy $\strFun \in \StrSet^{\must}(\sElm)$,
$(\abstr(\GName), \asgFun[\xElm \mapsto \strFun], \sElm \models^3 \psi) =\top$.
By induction, \WJnew{for all $\asgFun^*\in concr(\asgFun[\xElm \mapsto \strFun])$ and $t \in s$, it holds that $(\GName, \asgFun^*, t) \models^2 \psi$. 
Assume that $concr(\asgFun[\xElm \mapsto \strFun])$ is nonempty, and consider the sole concrete strategy $\strFun^* \in concr(\strFun)$. 
Clearly, $\asgFun^* = \asgFun^*[\xElm \mapsto \strFun^*]$ for every $\asgFun^*\in concr(\asgFun[\xElm \mapsto \strFun])$. Thus, $(\GName, \asgFun^*[\xElm \mapsto \strFun^*], t) \models^2 \psi$, and hence also $(\GName, \asgFun^*, t) \models^2 \EExs{\xElm}\psi$.
Assume now, to the contrary, 
that $concr(\asgFun[\xElm \mapsto \strFun])$ is empty. 
In that case, $(\GName, \asgFun^*[\xElm \mapsto \strFun^*], t) \models^2 \psi$ holds vacuously for all $\asgFun^* = \asgFun^*[\xElm \mapsto \strFun^*]$, and hence again $(\GName, \asgFun^*, t) \models^2 \EExs{\xElm}\psi$.}

Further, $(\abstr(\GName), \asgFun, \sElm \models^3 \phi) = \bot$ iff for every 
$\may$-strategy $\strFun \in \StrSet^{\may}(\sElm)$,
$(\abstr(\GName), \asgFun[\xElm \mapsto \strFun], \sElm \models^3 \psi) =\bot$.
Take any concrete strategy $g$ in $\GName$, and 
\WJnew{consider any $g^{\dagger} \in abstr^\may(\strFun)$. 
By the above, $(\abstr(\GName), \asgFun[\xElm \mapsto g^{\dagger}], \sElm \models^3 \psi) =\bot$. 
Thus, by induction, $(\GName, \asgFun^*, t) \not \models^2 \psi$ for all $\asgFun^*\in concr(\asgFun[\xElm \mapsto g^{\dagger}])$ and $t \in s$. 
Similarly to the previous paragraph, either (i) $concr(\asgFun[\xElm \mapsto g^{\dagger}])$ is nonempty and $\asgFun^*[\xElm \mapsto g] \in \asgFun^*$, thus $(\GName, \asgFun^*[\xElm \mapsto g], t) \not \models^2 \psi$ for all such $\asgFun^*$, or (ii) the same statement holds vacuously. In both cases, $(\GName, \asgFun^*, t) \not \models^2 \EExs{\xElm} \psi$.}

\smallskip
The cases \underline{$\phi = \forall{\xElm} \psi$} \WJnew{and \underline{$\phi = (\aElm, \xElm) \psi$}} are analogous.

\smallskip\noindent
\WJnew{\underline{Case $\phi = \X \psi$:} 
$(\abstr(\GName), \asgFun, \sElm \models^3 \phi) = \top$ iff 
for all $\playElm \in \plays^{\may}(\asgFun,\sElm)$, 
we have $(\GName, (\asgFun, \playElm)^{1}) \models^3 \psi) = \top$.
By induction, $(\GName, \asgFun^*,t) \models^2 \psi$ for every $\playElm \in \plays^{\may}(\asgFun,\sElm)$, $\asgFun^*\in concr(\asgFun_{\pi_{\le 1}})$ and $t \in (\pi)_{1}$.
Take any state $t'\in s$ and assignment $\asgFun'$ in $\GName$ such that 
$\asgFun'_{\pi^*_{\le 1}} = \asgFun^*$ for some $\playElm^* \in \plays(\asgFun',t')$.
Since \may paths in $\abstr(\GName)$ overapproximate paths in $\GName$, we get that $(\GName, \asgFun',t') \models^2 \X \psi$.

Further, $(\abstr(\GName), \asgFun, \sElm \models^3 \phi) = \bot$ iff 
for some $\playElm \in \plays^{\must}(\asgFun,\sElm)$, we have $(\GName, (\asgFun, \playElm)^{1} \models^3 \varphi) = \bot$.
By induction, there is $\playElm \in \plays^{\must}(\asgFun,\sElm)$ such that $(\GName, \asgFun^*,t) \not\models^2 \psi$ for every $\asgFun^*\in concr(\asgFun_{\pi_{\le 1}})$ and $t \in (\pi)_{1}$.
Take any state $t'\in s$ and assignment $\asgFun'$ in $\GName$. 
Since \must paths in $\abstr(\GName)$ underapproximate paths in $\GName$,  
there must be a path $\playElm^* \in \plays(\asgFun',t')$ such that $\asgFun'_{\pi^*_{\le 1}} = \asgFun^*$.
Thus, $(\GName, \asgFun',t') \not\models^2 \X \psi$.}

\smallskip
The cases \underline{$\phi = \psi_1 \U \psi_2$} and \underline{$\phi = \psi_1 \R \psi_2$} are analogous.
\end{proof}

\WJnew{
\begin{corollary}
For any CGS $\G$ and SL formula $\phi$: 
	\begin{eqnarray*}
	(\abstr(\GName) \models^3 \phi ) = \top
		& \Rightarrow & \G \models^2 \phi \\
(\abstr(\GName) \models^3 \phi ) = \bot
		& \Rightarrow & \G \not \models^2 \phi  
\end{eqnarray*}

\end{corollary}
}

\WJnew{It is easy to see that the above results hold also for the semantic variant of \SL based on memoryless strategies.}

\begin{table*}[t]
\scriptsize
\centering
\begin{tabular}{|r|rrrr|rrrrrr|}
\hline
\multicolumn{1}{|c|}{\multirow{2}{*}{Processes}} & \multicolumn{4}{c|}{CGS}                                                                                                                                                                                                                            & \multicolumn{6}{c|}{3-valued CGS}                                                                                                                                                                                                                                                                                                                                                                                                                                                                                          \\ \cline{2-11} 
\multicolumn{1}{|c|}{}                           & \multicolumn{1}{c|}{States} & \multicolumn{1}{c|}{Transitions} & \multicolumn{1}{c|}{\begin{tabular}[c]{@{}c@{}}Verification\\Time {[}sec{]} \\ (SL{[}1G{]})\end{tabular}} & \multicolumn{1}{c|}{\begin{tabular}[c]{@{}c@{}}Verification\\ Time {[}sec{]} \\ (SLK)\end{tabular}} & \multicolumn{1}{c|}{States} & \multicolumn{1}{c|}{\begin{tabular}[c]{@{}c@{}}Transitions \\ {[}Must{]}\end{tabular}} & \multicolumn{1}{c|}{\begin{tabular}[c]{@{}c@{}}Transitions \\ {[}May{]}\end{tabular}} & \multicolumn{1}{c|}{\begin{tabular}[c]{@{}c@{}}Abstraction\\ Time {[}sec{]}\end{tabular}} & \multicolumn{1}{c|}{\begin{tabular}[c]{@{}c@{}}Verification\\ Time {[}sec{]} \\ (SL{[}1G{]})\end{tabular}} & \multicolumn{1}{c|}{\begin{tabular}[c]{@{}c@{}}Verification\\ Time {[}sec{]}\\ (SLK)\end{tabular}} \\ \hline
2                                                & \multicolumn{1}{r|}{9}      & \multicolumn{1}{r|}{40}          & \multicolumn{1}{r|}{0.48}                                                                   & 0.18                                                                                 & \multicolumn{1}{r|}{5}      & \multicolumn{1}{r|}{16}                                                                & \multicolumn{1}{r|}{18}                                                               & \multicolumn{1}{r|}{0.01}                                                                 & \multicolumn{1}{r|}{0.10}                                                                                  & 0.03                                                                                               \\ \hline
3                                                & \multicolumn{1}{r|}{21}     & \multicolumn{1}{r|}{232}         & \multicolumn{1}{r|}{3725.56}                                                                & 2863.75                                                                              & \multicolumn{1}{r|}{6}      & \multicolumn{1}{r|}{24}                                                                & \multicolumn{1}{r|}{27}                                                               & \multicolumn{1}{r|}{0.03}                                                                 & \multicolumn{1}{r|}{0.12}                                                                                  & 0.13                                                                                               \\ \hline
4                                                & \multicolumn{1}{r|}{49}     & \multicolumn{1}{r|}{1376}        & \multicolumn{1}{r|}{T.O}                                                                    & T.O                                                                                  & \multicolumn{1}{r|}{7}      & \multicolumn{1}{r|}{34}                                                                & \multicolumn{1}{r|}{38}                                                               & \multicolumn{1}{r|}{0.10}                                                                 & \multicolumn{1}{r|}{0.24}                                                                                  & 0.19                                                                                               \\ \hline
5                                                & \multicolumn{1}{r|}{113}    & \multicolumn{1}{r|}{7904}        & \multicolumn{1}{r|}{T.O}                                                                    & T.O                                                                                  & \multicolumn{1}{r|}{8}      & \multicolumn{1}{r|}{46}                                                                & \multicolumn{1}{r|}{51}                                                               & \multicolumn{1}{r|}{0.31}                                                                 & \multicolumn{1}{r|}{0.72}                                                                                  & 0.69                                                                                               \\ \hline
6                                                & \multicolumn{1}{r|}{257}    & \multicolumn{1}{r|}{43520}       & \multicolumn{1}{r|}{T.O}                                                                    & T.O                                                                                  & \multicolumn{1}{r|}{9}      & \multicolumn{1}{r|}{60}                                                                & \multicolumn{1}{r|}{66}                                                               & \multicolumn{1}{r|}{1.24}                                                                 & \multicolumn{1}{r|}{2.08}                                                                                  & 1.12                                                                                               \\ \hline
7                                                & \multicolumn{1}{r|}{577}    & \multicolumn{1}{r|}{230528}      & \multicolumn{1}{r|}{T.O}                                                                    & T.O                                                                                  & \multicolumn{1}{r|}{10}     & \multicolumn{1}{r|}{76}                                                                & \multicolumn{1}{r|}{83}                                                               & \multicolumn{1}{r|}{10.88}                                                                & \multicolumn{1}{r|}{5.66}                                                                                  & 4.99                                                                                               \\ \hline
8                                                & \multicolumn{1}{r|}{1281}   & \multicolumn{1}{r|}{1182208}     & \multicolumn{1}{r|}{T.O}                                                                    & T.O                                                                                  & \multicolumn{1}{r|}{11}     & \multicolumn{1}{r|}{94}                                                                & \multicolumn{1}{r|}{102}                                                              & \multicolumn{1}{r|}{107.89}                                                               & \multicolumn{1}{r|}{8.40}                                                                                  & 6.67                                                                                               \\ \hline
9                                                & \multicolumn{1}{r|}{2817}   & \multicolumn{1}{r|}{5903872}     & \multicolumn{1}{r|}{T.O}                                                                    & T.O                                                                                  & \multicolumn{1}{r|}{12}     & \multicolumn{1}{r|}{114}                                                               & \multicolumn{1}{r|}{123}                                                              & \multicolumn{1}{r|}{1087.14}                                                              & \multicolumn{1}{r|}{29.37}                                                                                 & 26.81                                                                                              \\ \hline
\end{tabular}
\caption{Experimental results for the scheduler case study (T.O. stands for Time Out).}
\label{tab:experiments}

\end{table*} 
\section{Implementation}
\label{sec:implementation}

We implemented a prototype tool in Java\footnote{\url{https://github.com/AngeloFerrando/3-valuedSL}}, which accepts CGSs and SL properties as input,
on top of MCMAS, the \textit{de facto} standard model checker for MAS~\cite{LomuscioQuRaimondi15}. Specifically, our tool exploits MCMAS as a black-box, for performing the actual verification step. In fact, our tool focuses on the abstraction procedure for the verification of SL formulas (as presented in this paper), while their verification is obtained through MCMAS.

From a practical perspective, there are various aspects to report, that can be summarised as
(i) input/output of the tool;
(ii) abstraction of the CGS;
(iii) verification in MCMAS.

(i) The implementation allows for the definition of CGSs as external JSON\footnote{\url{https://www.json.org/}} formatted input files. In this way, any end user may easily interact with the tool, independently from
the CGS's internal representation (\textit{i.e.}, the corresponding data structures).
As CGSs, also the definition of the SL formula to
check is handled as an external parameter to the tool. Once the verification ends,
the outcome is returned to the user.

(ii) As presented in the paper, in order to improve the verification performance, the CGS is first abstracted. The abstraction is obtained by clustering multiple states into a single abstract state of the CGS. This step is based on an equivalence relation ($\approx$), as presented in Definition~\ref{def:abstraction}.
An abstract state may be labeled by atoms.
As presented in Definition~\ref{def:abstraction}, an atom holds (resp.~does not hold) in the abstract state iff it holds (resp.~does not hold) in all the concrete states which have been collapsed into the abstract state. Otherwise, the atom is considered undefined.
Note that, since atoms can hold, not hold, or being undefined in a state, they are explicitly labeled in each state. In practice, this is obtained by duplicating each atom $p$ into atoms $p_\top$ and $p_\bot$,
which correspond to $p$ holding or not holding in a certain state of the abstract CGS; whereas being undefined can be marked by having neither $p_\top$ nor $p_\bot$ present in the abstract state.

(iii) The abstract CGS is then verified in MCMAS against an SL formula. In more detail, our tool exploits the MCMAS extension for SL[1G], i.e., the \textit{one goal} fragment \cite{CermakLomuscioMurano15}, and the MCMAS extension for SLK, i.e., an epistemic extension of SL, \cite{Cermak+14}.

Note that, to make use of the MCMAS model checker, our CGSs need to be first translated into Interpreted Systems~\cite{FHMV95}. In fact, MCMAS does not support CGSs, and it expects Interpreted Systems expressed using a domain specific language called Interpreted Systems Programming Language (ISPL).
Thus, a pre-processing step before calling MCMAS is always required, where the CGS of interest is first automatically translated into its ISPL representation. This is only a technical detail, since CGSs and Interpreted Systems are equally expressive \cite{BLMR20,Goranko+04a}.

It is important to report that the ISPL generation is performed on standard CGSs, not on their abstraction. Indeed, the abstract CGSs as described in Definition~\ref{def:abstraction} cannot be used in MCMAS straight away, but need to be reprocessed first.
To generate a CGS which can then be verified into MCMAS, the tool splits the 3-valued CGS into two CGSs. Such a split is determined by the SL formula under evaluation;
\nmnew{that is, given an SL formula $\varphi$, we extract two sets of agents, $E$ and $U$, whose strategies are only existentially and universally quantified in $\varphi$, respectively.}
By using these two sets, we split the 3-valued CGS into two CGSs:
one CGS where agents in $E$ use $must$-strategies, while agents in $U$ use $may$-strategies;
one CGS where agents in $E$ use $may$-strategies, while agents in $U$ use $must$-strategies.
The first CGS can be used to prove the satisfaction of $\varphi$, while the second CGS can be used to prove the violation of $\varphi$. This follows from Definition~\ref{def:sl(semantics)}, third and fourth bullet points.

As a consequence of how the verification is performed in practice, we remark an important difference between the theory presented in this paper and its implementation:
the implementation handles SL formulas with arbitrary alternation of universal ($\forall x$) and existential ($\exists x$) quantifiers, as long as for each agent ($a$) in the formula, there is one single binding $(a, x)$.
Even though at the theoretical level our abstraction method can handle all SL formulas, at the implementation level this is not the case.
In fact, our tool is based on MCMAS, and because of that, we cannot handle formulas where the agents need to swap between universally and existentially quantified strategies.
This would require to modify MCMAS internally, which we leave as future work.

\begin{figure}
    \centering
    \includegraphics[width=.63\linewidth]{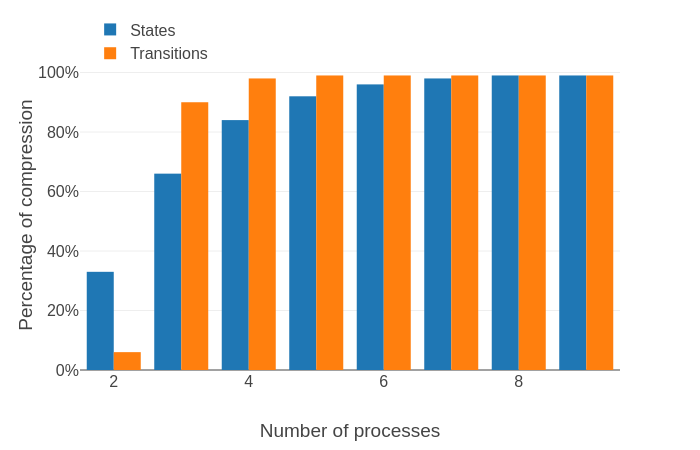}
    \caption{CGS compression in the scheduler case study.}
\label{fig:compression}
\end{figure}

\section{Experiments}
\label{sec:experiments}

We carried out the experiments on a machine with the following specifications: Intel(R) Core(TM) i7-7700HQ CPU @ 2.80GHz, 4 cores 8 threads, 16 GB RAM DDR4.
The case study we experimented on consists in a scheduler, where $N$ agents, \textit{i.e.}, processes (called $P_i$ for $1 \leq i \leq N$) compete to get CPU time, while an $Arbiter$ agent decides which process to grant access (one at a time).
The full description of the example can be found in~\cite{DBLP:journals/iandc/CermakLMM18}.
The corresponding CGS can be parameterised over the number $N$ of processes. Naturally, such parameter largely influences the size and complexity of the resulting CGS.
Table~\ref{tab:experiments} reports experimental results we obtained by applying our tool to the scheduler case study. We considered the verification of the same SL formula $\varphi$ verified in~\cite{DBLP:journals/iandc/CermakLMM18}, that is:
{\footnotesize
\begin{center}
$\varphi = \forall x, \vec{y} (Arbiter,x) (P_1,y_1) \ldots (P_n,y_n) \mbox{}$\\
$\qquad G \lnot \bigvee_{i=1}^{n}\bigvee_{j=i+1}^{n} rs_i \land rs_j$
\end{center}
}
Intuitively,
$\varphi$ asserts that at most one
process ($P_i$) owns the resource ($rs$) at any given point in time.
In Table~\ref{tab:experiments}, each row refers to
a fixed number of processes, from 2 to 9, used to generate the corresponding CGS.
Each row also
reports the number of states and transitions of the CGS, and the time required to perform its verification in MCMAS,
both on the original CGS and its 3-valued abstraction,
for comparison.
For the latter, the time required to generate such an abstraction is reported.
For the experiments with the scheduler, the abstraction is assumed to be guided by an expert of the system. In more detail, all states where at least one process is waiting to be selected by the scheduler are clustered together.
This choice, as apparent in Table~\ref{tab:experiments}, largely reduces the number of states and transitions of the CGS. Nonetheless, this does not prevent the verification process to correctly conclude the satisfaction of $\varphi$ on both the CGS and its 3-valued version, \textit{i.e.}, the abstraction does not remove any information necessary to determine the satisfaction of $\varphi$.
Table~\ref{tab:experiments} also reports the execution time required for the actual verification of both the CGS and its 3-valued abstraction. As we can observe, without the abstraction step, the verification of the CGS times out when 3 processes are considered. In fact, MCMAS cannot model check $\varphi$ in less than 3 hours, which was set as the time out (both for the SL[1G] and SLK extensions of MCMAS). Instead, thanks to the abstraction, the verification can be performed for up to 9 (a more realistic number of processes). Note that, the verification of the 3-valued CGS could have been performed for even larger numbers of processes. However, the CGS with 10 processes did not fit into the available memory of the machine used for the experiments; so, it was not possible to apply our technique to generate its 3-valued abstraction. Nonetheless, we expect the tool to handle even the case with $10$ processes via abstraction.
Figure~\ref{fig:compression} reports the data compression obtained in the scheduler case study. It is immediate to observe the huge compression obtained via abstraction. Indeed, the larger is the number of processes involved, the more significant is such compression. Note that, for more than 6 processes, the abstra\-ction produces a CGS with $\sim\!99$\% less states and transitions.
Besides $\varphi$, we experimented with other specifications as well. Specifically, we carried out experiments over a large set of randomly-generated SL formulas.
The goal of these experiments is to understand how many times our tool would return a conclusive answer (\textit{i.e.}, not $\uval$). We automatically synthesised 10,000 different SL formulas and verified them in the scheduler case study; where we kept the same abstraction as for Table~\ref{tab:experiments}. Over the 10,000 different SL formulas, the tool was capable of providing a defined truth value (either true or false) in the 83\% of cases. Of course, this is a preliminary evaluation, which needs to be corroborated through additional experiments, also involving further real-world scenarios. Nonetheless, the results we obtained are promising, and allow us to empirically show the effectiveness of our approach, not only from a data-compression perspective, but also from a computational one.

\section{Conclusion}
\label{sec:conclusion-future}

The high complexity of the verification problem for Strategy Logic hinders the development of practical model checking tools and therefore its application in critical, real-life
scenarios.  As a consequence, it is of upmost importance to develop
techniques to alleviate this computational burden and allow the use of
Strategy Logic in concrete use cases, such as the scheduler scenario
here analysed.  This contribution is meant to be the first step in this
direction.

\newpage
\section*{Acknowledgements}
W. Jamroga acknowledges the support of NCBR Poland, NCN Poland, and FNR Luxembourg under projects STV (POLLUX-VII/1/2019), SpaceVote (POLLUX-XI/14/SpaceVote/2023), and SAI (2020/02/Y/ST6/00064).
A. Murano acknowledges
the support of the PNNR FAIR project, the InDAM project “Strategic Reasoning in Mechanism Design”, and the PRIN 2020 Project RIPER.

\bibliographystyle{named}
\newcommand{\hoek}[1]{}

\end{document}